\documentclass[a4paper,12pt]{article} 
\usepackage{amsfonts,amssymb,amsthm,amsmath}
 \usepackage[dvips]{graphics}
\usepackage{color}

\DeclareSymbolFont{AMSb}{U}{msb}{m}{n}

\begin{document}


\def\N{\mathbb N}
\def\M{\mathbb M}
\def\fdp{f\circ d_{p}}
\def\R{\mathbb R}
\def\DE{D(\mathcal E)}
\def\DEC{\mbox{D}_c(\mathcal E)}
\def\mint{\,\mathop{\makebox[0pt][c]{$\int$}\makebox[0pt][c]{--}\,}}
\def\mintl{\!\!\mathop{\makebox[0pt][c]{$\int$}\makebox[0pt][c]{--}\!}\limits}
\def\suml{\sum\limits}
\def\Er{\mathcal E^r}
\def\E{\mathcal E}
\def\Erw{\mathbb E}
\def\intl{\int\limits}
\newcommand{\Br}[1]{B_r(#1)}
\newcommand{\Bbr}[1]{\mathbb B_r(#1)}
\newcommand{\bnk}[1]{b_{n,k}(#1)}
\newcommand{\bn}[1]{b_{n}(#1)}
\def\phx{\phi_x}
\def\Mkx{\M_k(x)}
\def\bm{{\mathop{\mbox{\normalfont \,vol}}}}
\def\dbm{{{\mbox{\normalfont \tiny d \!\!vol}}}}
\def\dbmn{{{\mbox{\normalfont d \!\!vol}}}}
\def\Qrx{q_r(x)}
\def\Regset{X\setminus S_X}
\def\volsing{\mathbb S_X}
\def\angle{\sphericalangle}
\def\Eomega{\E_\Omega}
\def\Deomega{\mathcal D(\Eomega)}
\def\Meas{\cal M}
\def\1{\,{\makebox[0pt][c]{\normalfont    1}
\makebox[2.5pt][c]{\raisebox{3.5pt}{\tiny {$\|$}}}
\makebox[-2.5pt][c]{\raisebox{1.7pt}{\tiny {$\|$}}}
\makebox[2.5pt][c]{} }}
\def\one{\1 }
\def\eI{[0,1]}

\def\one{\1 }
\def\B{\mathcal B}
\newcommand{\infl}[1]{\inf\limits_{#1}}
\newcommand{\gsk}[1]{\left\{ #1 \right\}}
\newcommand{\norm}[1]{\left\| #1 \right\|}
\def\P{{ /\!\!/}}
\def\ul{\underline}
\def\H{\mathcal H}
\def\for{\mbox{ for }}
\def\Ricc{\mathop{\mbox{\normalfont Ricc}}}
\def\Sec{\mathop{\mbox{Sec}}}

\def\sth{\,|\,}
\def\Hess{\mathop{\mbox{{\normalfont Hess}}}}
\def\einh{\frac{1}{2}}
\def\ol{\overline}
\def\abs{~\\~\\}
\def\sabs{~ \smallskip ~}
\def\k{,\,}
\def\supp{\mathop{\mbox{supp}}}
\def\Curv{\mbox{Curv}}

\def\cut{\mbox{\normalfont Cut}}
\def\mt{\rightarrow}
\newcommand{\lgauss}[1]{{\lfloor #1 \rfloor}}
\def\falls{\mbox{ if }}
\def\sonst{\mbox{ else}}
\def\grad{\mbox{grad}}
\def\div{{{\mathop{\,{\rm div}}}}}
\def\diam{{{\mathop{\,{\rm diam}}}}}
\def\Id{{\mathop{{\bf 1}_{\small \R^d}}}}
\def\k{,\,}
\def\dist{{\mbox{dist}}}
\def\Ind{{\mathop{\mbox{I}}}}
\def\dil{{\mathop{\mbox{dil}}}}
\def\H{\mathbb H}
\def\mtm{M \times M}
\def\BV{\mbox{\normalfont BV}}
\def\gij{{g_{ij}}}
\def\gije{{g_{ij}^\epsilon}}
\def\Lip{\mathop{\mbox{\normalfont Lip}}}
\def\lipc{\mathop{\mbox{\textit {lip}}}}
\def\Lipc{\mathop{\mbox{\textit {Lip}}}}
 \def\loc{{\mbox{\scriptsize loc}}}
\renewcommand{\bullet}{{\mathbf \cdot  }}
\def\Cut{\mbox{\normalfont Cut}}
\def\dist{\mbox{\normalfont dist}}
\def\sign{\mbox{\normalfont sign}}
\def\pr{\mbox{\normalfont pr}}

\def\Re{{\rm Re}}

\def\Im{{\rm Im}}

\def\ioe{\frac{1}{\epsilon}}
\def\fe{f_{\epsilon}}
\def\Ebr{\E^{b,r}}
\def\Dco{D_c(\E_\Omega)}

\newcommand{\D}{{\mathbb{D}}}
\newcommand{\T}{{\mathbb{T}}}
\newcommand{\DD}{\overline{\mathbb{D}}}
\newcommand{\LL}{{\mathbb{L}}}
\newcommand{\EE}{{\mathbb{E}}}

\newenvironment{bew}{\begin{proof}}
{\end{proof}}

\newcommand\iitem{\stepcounter{idesc}\item[(\roman{idesc})]}
\newenvironment{idesc}{\setcounter{idesc}{0}
\begin{description}} {\end{description}}

\def\dnvol{{d \ol{\normalfont v}}}

\def\P{{\mathcal P}}
\def\Z{{\mathcal Z}}
\def\CMS{CM[0,1]^2}
\def\eI{{[0,1]}}
\def\L2{{L^2}}
\def\L11{\mathcal P _{ac}}
\def\S1{{S^1}}
\def\llangle{\langle \langle}
\def\rrangle{\rangle \rangle}

\newcounter{zaehler}
\setcounter{zaehler}{1}
\newcommand{\inszahl}{{\normalfont\tiny\thezaehler.}\addtocounter{zaehler}{1}}

\def\Pbeta{\mathbb P^\Beta}
\def\G{\mathcal G}
\def\Qbeta{{\mathbb Q^\beta}}
\def\Qnull{{\mathbb Q^0}}
\def\dreh{\mathcal R}
\def\Pbeta{\mathbb P^\beta}
\def\ssurd{{\hfill $\surd$}}
\def\bbox{{\hfill $\Box$}}
\def\Def{{\bf  \inszahl Definition.~}}
\def\Satz{{\bf \inszahl Satz. ~}}
\def\Lem{{\bf \inszahl Lemma.~}}
\def\Bew{{\it Beweis. ~}}
\def\Kor{{\bf \inszahl Korollar. ~}}
\def\Top{~\\{\bf \inszahl\,}}
\def\id{{\mbox{id}}}
\def\L{{\mathcal L}}
\def\Lbeta{{\mathcal L^\beta}}

\def\intfgamma{{\mbox{$\int\!F(\gamma)$}}}
\def\intfg{{\mbox{$\int\!F(g)$}}}
\def\intggamma{{\mbox{$\int\!G(\gamma)$}}}
\def\inthg{{\mbox{$\int\!H(g)$}}}
\def\nN{\mathbb N}
\def\H{{\mathcal H}}
\def\C{{\mathcal C}}
\def\W{{\mathcal W}}
\def\D{{\mathcal D}}
\def\smint{{\mbox{$\int$}}}
\def\sabs{ \\ \smallskip}
\def\inthg{{\mbox{$\int\!H(g)$}}}
\def\rinv{{\rm rinv}}
\def\remarkend{}
\def\dnull{{d_\E}}

\newcommand{\leb}{{\mbox{Leb}}}
\newcommand{\Cyl}{{\mathfrak{C}}}
\newcommand{\Syl}{{\mathfrak{S}}}
\newcommand{\Zyl}{{\mathfrak{Z}}}
\def\smint{{\mbox{$\int$}}}

\newcommand{\MCP}{\mbox{{\sf{MCP}}}}
\newcommand{\CD}{\mbox{{\sf{CD}}}}
\newcommand{\Length}{\mbox{\rm Length}}
\newcommand{\Ent}{\mbox{\rm Ent}}
\newcommand{\Gaps}{\mbox{\rm gaps}}

\newcommand{\Dom}{{\it{Dom}}}
\newcommand{\tr}{\mbox{\rm tr}}
\newcommand{\X}{{\mathbb{X}}}
\newcommand{\A}{{\mathcal{A}}}
\newcommand{\U}{{\mathcal{U}}}
\newcommand{\V}{{\mathcal{V}}}
\newcommand{\inv}{\mbox{\sf inv}}
\newcommand{\cum}{\mbox{\sf cum}}
\newcommand{\varh}{{h}}
\newcommand{\Pz}{{\mathcal{P}_2}}
\newcommand{\Pe}{{\mathcal{P}}}
\newcommand{\riccurv}{\,\underline{\mathbb{C}\mbox{\rm urv}}}
\newcommand{\alexcurv}{\,\underline{\mbox{{\sf curv}}}}

\newtheorem{theorem}{Theorem}[section]
\newtheorem{lemma}[theorem]{Lemma}
\newtheorem{definition}[theorem]{Definition}
\newtheorem{corollary}[theorem]{Corollary}
\newtheorem{proposition}[theorem]{Proposition}
\newtheorem{remark}[theorem]{Remark}
\newtheorem{remarks}[theorem]{Remarks}
\newtheorem{example}[theorem]{Example}

\newcommand{\Q}{{\mathbb{Q}}}
\newcommand{\Pp}{{\mathbb{P}}}

\renewcommand{\labelenumi}{(\roman{enumi})}

\def\proof{\smallskip \noindent{\textit{Proof.~}}}

\begin{center}
\vspace{-2em}
 {\large AN OPTIMAL TRANSPORT  VIEW ON SCHR\"ODINGER'S EQUATION}

\smallskip
\smallskip {Max-K. von Renesse
\def\thefootnote{} \footnote{
February  27 2009, Technische Universit\"at Berlin, email:
mrenesse\@@math.tu-berlin.de, \textbf{Keywords} Schr\"odinger
Equation, Optimal Transport, Newton's law, Symplectic Submersion.
 } }

\bigskip
{\small \bf Abstract}

\bigskip

\begin{minipage}{16.3cm}
{\small 
We show that the Schr\"odinger equation  is a lift of Newton's law
of motion $\nabla^\W_{\dot \mu} \dot \mu = -\nabla^\W F(\mu)$  on
the space of probability measures, where derivatives are taken
w.r.t. the Wasserstein Riemannian metric. Here the potential
  $\mu \to F(\mu)$ is the 
%
sum of the total classical potential energy $\langle V,\mu\rangle$
of the extended system
 and  its  Fisher information
  $  \frac {\hbar^2} 8 \int |\nabla \ln \mu |^2
d\mu$. The precise relation is established via a well known
('Madelung') transform which is shown to be a symplectic submersion
of the standard symplectic
 structure of complex valued functions into the
canonical symplectic space over  the Wasserstein space. 
All computations are conducted in the framework of Otto's formal
Riemannian calculus for optimal transportation of probability
measures. }
\end{minipage}

\end{center}

\def\H1{{\mathbb H^1}}




\section{\normalsize \normalfont INTRODUCTION}
 Recent applications  of optimal transport theory have demonstrated
that certain analytical and geometric problems on finite dimensional
Riemannian manifolds $(M, g)$ or more general metric measure spaces
$(X,d,m)$ can nicely be treated in the corresponding ('Wasserstein')
space of probability measures $\P_2(X) = \{ \mu \in \P(X)\,|\,
\int_X d^2(x,o) \mu(dx) <\infty\}$ equipped with  the quadratic
Wasserstein metric
\[d_\W (\mu, \nu) = \inf \left\{ \iint_{X^2} d^2(x,y)
\Pi(dx,dy)\,\left |\, \Pi \in \mathcal P(X^2), \Pi(X\times A) = \nu,
\Pi(A\times X) =\mu(A),A \in \mathcal B (X)\right. \right\}^{1/2}.\]
This metric corresponds to a relaxed version of Monge's  optimal
transportation problem with cost function   $c(x,y)= d^2(x,y)$
\[ \inf \left\{ \int_{X} c(x,Ty) \mu(dx) \,\left |\, T: X\to X, T_*\mu =\nu\right.\right\},\]
with $T_*\mu$ denoting the image (push forward) measure of $\mu \in \mathcal P(X) $ under the map $T$.\smallskip

The physical relevance of the Wasserstein distance was highlighted
by the works of e.g. BENAMOU-BRENIER \cite{benamou-brenier} and OTTO
\cite{MR1842429} who established  in  the  smooth Riemannian case
$X=M$ and smooth  initial distribution $\mu$
\[
d_\W^2 (\mu, \nu) = \inf \left\{ \int _0^1 \int_M |\nabla
\phi_t(x)|^2 \mu_t(dx)dt \,\left |\,
\begin{array}{cc}
 \phi\in C^\infty(]0,1[\times M ), t \to \mu_t \in C([0,1], \mathcal P(M))\\
 \dot \mu_t = - \div(\nabla \phi_t   \mu_t ), t\in ]0,1[, \mu_0 =\mu, \mu_1 =\nu
\end{array}
 \right\}\right. ,\]
showing that   $d_\W$ is associated to a  formal Riemannian
structure  on $ \mathcal P (M) $ given by
\[
   T_\mu\P(M) =\{ \psi :M \to \R, \int_M \psi (x) dx=0 \},\]
\[ \norm{\psi}^2_{T_\mu\P}  = \int_M |\nabla \phi |^2d\mu,   \mbox{ for } \psi = -\div(\mu \nabla\phi).\]
In view of the continuity equation
\[  \dot \mu_t = -\div( \dot \Phi_t \mu_t) \]
for a smooth flow $(t,x) \to \Phi_t(x)$  on $M$, acting on measures
$\mu$ through push forward $\mu_t = (\Phi_t)_*\mu_0$,  this
identifies the Riemannian energy of a curve $t\to \mu_t \in \P(M)$
with the minimal required  kinetic energy
\[ E_{0,t}(\mu)  = \int _0 ^t \norm{\dot \mu_s}_{T_{\mu_s}\P(M)} ^2 ds = \int_0^t \int_M   |\dot \Phi(x,s) |^2 \mu_s(dx)ds.\]
A major reason for the  success of this framework is the
interpretation of evolution equations of type
\[ \partial_t   u = \div (u_t
\nabla F'(u)), \] with   $F'$ being the $L^2$-Frechet  derivative of
some smooth functional $F$ on $L^2(M, dx)$, as $d_\W$-gradient
('steepest descent') flow \begin{equation*} \dot \mu = -\nabla^\W F
(\mu) \label{gradflow} \end{equation*} for the measures $\mu (dx) =
u(x) dx$. Properties of the flow my thus be deduced from the
geometry  of the funtional $F$ with respect to $d_\W$. A
particularly important case is the Boltzmann entropy $F(u) = \int_M
u \ln u \,dx$ which induces the heat flow.
\\

In this note we propose an example of another natural class of
dynamical systems associated
with the Riemannian metric on $\P(M)$ and which can be written as   
\begin{equation}
\nabla^\W_{\dot \mu} \dot \mu =  - \nabla^\W F(\mu).
 \label{eulerlagrange}
\end{equation}
Equation \eqref{eulerlagrange} describes the Hamiltonian flow on
$T\P(M)$ induced from the Lagrangian
\[
 L_F : T\P(M) \to \R; \quad L_F(\psi)= \frac 1 2 \norm{\psi}^2_{T_\mu\P} - F(\mu) \quad \mbox{ for }  \psi \in T_{\mu}\P(M) \]
with the functional   $F: \P(M)\to \R$ now playing the role of a
potential field for the system. Apart from the closely related
recent work \cite{gnt} it seems that a systematic approach to such
Hamiltonian flows on $\P(M)$ is missing in the literature. The
example we want to propose is obtained by  choosing
\begin{equation} F(\mu) =  \int_M V(x) \mu(dx) + \frac {\hbar^2} 8 I(\mu), \label{genpot}
\end{equation}
where
\[ I(\mu) =  \int_M |\nabla \ln \mu|^2 d\mu. \]
We show that via an appropriate transform the flow
\eqref{eulerlagrange} solves the Schr\"odinger equation
\begin{equation}
 i \hbar \partial _t \Psi = - \hbar^2 /2    \Delta  \Psi + \Psi \, V.  \label{seq}
\end{equation}
The functional $I$ is known today as Fisher information. Physically
$I(\mu)$ is the instantaneous kinetic energy required by the
unperturbed heat flow at state $\mu$. The prominent role of $I$ for
quantum behaviour was noticed long ago, e.g.\ in a classical paper
by BOHM \cite{MR0046287}, using the following well-known system of
generalized Hamilton-Jacobi and transport equations
\begin{equation}
 \begin{split}
 \partial _t   S+ \frac  1 2 |\nabla   S|^2  & + V + \frac {\hbar^2} 8 \bigl ( | \nabla \ln \mu |^2 - \frac 2 \mu \Delta \mu \bigr) =0  \\
 \partial _t\mu & + \div (\mu \nabla   S  )=0.
 \end{split} \label{madelungflow}
\end{equation}
This system  was poposed by MADELUNG very early \cite{madelung} as
 an equivalent description of the wave function $\Psi= \sqrt \mu e ^{\frac i \hbar S}$
under the Schr\"odinger equation. In the sequel it will be referred
to as Madelung flow. Various attempts to derive it from first order
principles can be found in the physics literature, e.g.\ most
recently in \cite{MR1914637}.\\

Our present note starts with the observation that equations
\eqref{madelungflow} and
 \eqref{eulerlagrange} are essentially the same (theorem \ref{mainthm}),  where the latter is  understood in the sense of LOTT's recently proposed
  second order calculus on
Wasserstein space, c.f.\  \cite{MR2358290}.
%
%
A virtue of formula \eqref{eulerlagrange} is its very intuitive
physical interpretation as  Newton's law for the motion of an
extended system with inertia (we have put mass density equal to
one). Acceleration comes from a a gradient field of a potential $F$
which is the total mechanical potential of the extended system plus
its 'kinetic potential' w.r.t.  the heat flow.  (As usual the case
of a classical single particle moving in a potential field is
embedded naturally in
 \eqref{eulerlagrange} if one puts $\hbar =0$ and
$\mu=\delta_{x}$.)\\

Secondly  we show that the  two equations  \eqref{eulerlagrange} and
\eqref{seq} are, modulo constant phase shifts, symplectically
equivalent. More precisely, we compute the canonical symplectic form
on the tangent bundle $T\P(M)$ induced from the Levi-Civita
connection of the Wasserstein metric on $\P(M)$ and  show that the
map $\Psi= |\Psi| e^{\frac i h S} \mapsto -\div (|\Psi|^2 \nabla
S)$, which we shall call Madelung transform,  is a symplectic
submersion of the standard Hamiltonian structure of the
Schr\"odinger equation on the space of complex valued functions into
the Hamiltonian structure associated to \eqref{eulerlagrange} on the
tangent bundle $T\P(M)$. Except for its curiosity in Wasserstein
geometry this result seems to support the point of view of some
authors  that the familiar complex valued form (\ref{seq}) of the
Schr\"odinger equation is the consequence of a smart choice of
coordinates in which the intuitive but
unhandy dynamical system \eqref{eulerlagrange} resp. \eqref{madelungflow} can be solved very efficiently. \\

%

Obviously, much of what is presented below resembles the familiar
Schr\"odinger folklore, c.f.\ in particular NELSON's theory of
stochastic mechanics \cite{MR783254} and its follow-ups, e.g.\
\cite{MR2165828}. And in fact nothing really new about the
Schr\"odinger equation itself is implied at this point. Our aim is
the connection to Wasserstein geometry which in our view  gives a
very intuitive picture. Finally, we  emphasize that all of our
computations are completely formal, a rigorous mathematical
treatment of these ideas is subject to future work.

\section{\normalsize \normalfont SCHR\"ODINGER EQUATION FROM NEWTON's LAW OF MOTION ON $(\P(M),d_\W)$}
The computations below are conducted on the formal Riemannian
manifold of fully supported smooth probability measures equipped
with the Wasserstein metric tensor, as initiated in
\cite{MR1842429,MR1760620} and extended in \cite{MR2358290},
ignoring full mathematical generality or rigor. (The basic
background material taken from \cite{MR2358290,MR1842429} can be
found in the appendix.) In the sequel we shall often identify $\mu
\in \P^\infty(M)$ with its density  $\mu \stackrel {\wedge}{=}
d\mu/dx$.

\begin{theorem}\label{mainthm}\textit{For $V \in C^\infty (M)$ let $F: \P^\infty (M) \to \R$  defined as in (\ref{genpot}).
Then any smooth local solution $t\to \mu(t) \in \P(M)$ of
\eqref{eulerlagrange}
%
yields a local solution $(\mu_t, \overline S_t)$ of  the Madelung
flow \eqref{madelungflow}, where
\[ \bar S(x,t) = S(x,t) + \int_0^t L_F(S_\sigma,\mu_\sigma)d\sigma  \]
and $S(x,t) $ is the velocity potential of the flow $\mu$, i.e.
satisfying $\int _M S d\mu =0$ and $\dot \mu_t = -\div (\nabla
S_t\mu)$. Conversely, let  $(\mu_t , S_t)$ be smooth a local
solution of \eqref{madelungflow} then   $t\to\mu_t \in \P(M)$ solves
\eqref{eulerlagrange}.}
\end{theorem}

\noindent
\begin{proof}
 Let    $\mu$ solve (\ref{eulerlagrange})
where $\nabla^\W$ is the Wasserstein gradient and $ \nabla^\W_{\dot \mu} \dot \mu$  is the covariant derivative  associated to the Levi-Civita connection on $T\P(M)$. Let $(x,t) \to S(x,t)$ denote the velocity potential  of $\dot \mu$ (cf. section \ref{append}), then according to \cite[proposition 4.24]{MR2358290}
the left  hand side of (\ref{eulerlagrange}) is computed as
\[
 - \div\left(\mu \nabla \left( \partial _t S + \frac 1 2 |\nabla S|^2\right) \right),
\]
where the right hand side of (\ref{eulerlagrange}) equals (cf. section \ref{append})
\[ \div\left( \mu \nabla \left(   V + \frac {\hbar^2} 8 \bigl(  |\nabla \ln \mu|^2 - \frac 2 \mu \Delta \mu\bigr)\right)\right).
\]
Since $\mu_t $ is fully supported on $M$ this implies
\[ \partial _t S + \frac 1 2 |\nabla S|^2+ V +  \frac {\hbar^2} 8  \bigl( |\nabla \ln \mu|^2 - \frac 2 \mu \Delta \mu \bigr)= c(t)\]
for some function $c(t)$. To compute $c(t)$ note that due to the normalization $\langle S_t, \mu_t\rangle=0$
\begin{align*}
 0 & = \partial _t \langle S_t, \mu_t\rangle \\
& = c(t) - \frac 1 2 \langle |\nabla S|^2 , d\mu\rangle - F(\mu) + \langle S, \dot \mu \rangle \\
&  = c(t) - \frac 1 2 \langle |\nabla S|^2 , d\mu\rangle - F(\mu)  + \langle |\nabla S |^2 , \mu \rangle = c(t) + L_F(S_t, \mu_t).
\end{align*}
Hence the pair $t\to (  \ol S_t,\mu_t)$ with $\bar S(x,t) = S(x,t)
+\int_0^t L_F(S_\sigma,\mu_\sigma)d\sigma $ solves
\eqref{madelungflow}. The converse statement is now also
obvious.\bbox
 \end{proof}

\begin{corollary}\label{schroedcor}\textit{For $V \in C^\infty (M)$ let $F: \P^\infty (M) \to \R$  defined as in (\ref{genpot}).
Then any smooth local solution $t\to \mu(t) \in \P(M)$ of
\begin{equation*}
\nabla^\W_{\dot \mu} \dot \mu =  - \nabla^\W F(\mu),
\end{equation*}
%
yields a local solution of  the Schr\"odinger equation \eqref{seq}
via
\begin{equation}  \Psi(t,x) = \sqrt {\mu(t,x)} e^{\frac i \hbar   \bar S(x,t) }
\label{invmadtrans}
\end{equation}
where
\[ \bar S(x,t) = S(x,t) + \int_0^t L_F(S_\sigma,\mu_\sigma)d\sigma  \]
and $S(x,t) $ is the velocity potential of the flow $\mu$, i.e.
satisfying $\int _M S d\mu =0$ and $\dot \mu_t = -\div (\nabla
S_t\mu)$. }
\end{corollary}

\begin{remark} \label{factorize}{\normalfont The passage from $S$ to $\ol S= S + const.$
does not bear any physical relevance, since two wave functions
$\Psi, \tilde \Psi $ with $\tilde \Psi = e^{ i \kappa} \Psi $ for
some $\kappa \in \R$ parameterize the same physical system.
Accordingly the Schr\"odinger equation should probably rather be
understood in the sense of  $ i \hbar
\partial _t [\Psi] = - \hbar^2 /2    \Delta [\Psi] + [\Psi] \, V
$ for a flow of equivalence classes of wave functions. On the level
of representatives this amounts to the equation
\[ \exists \, \kappa(.) :\R_+ \to \R: \quad  i \hbar
\partial _t \Psi = - \hbar^2 /2    \Delta  \Psi + \Psi \, V +  i
\kappa \Psi. \]
} \end{remark}

\begin{remark}{\normalfont The $d_\W$-gradient flow on $\P(M)$ for $F$ as in
\eqref{genpot} corresponding to the overdamped limit of
\eqref{eulerlagrange} gives  a nonlinear 4th-order equation which is
sometimes called the 'Derrida-Lebowitz-Speer-Spohn' or
'quantum-drift-diffusion' equation. A rigorous treatment of it can
be found in \cite{GST}. }
\end{remark}

The usual argument for the derivation of Euler-Lagrange equations
yields the following statement. \smallskip

\begin{corollary}\label{eulerlaglem}\textit{For $V \in C^\infty (M)$ let $F: \P^\infty (M) \to \R$  defined as in (\ref{genpot}). Then any smooth local Lagrangian flow  $[0,\epsilon ] \ni t \to  \dot \mu_t\in T\P^\infty(M)$  associated to $L_F$
yields a local solution of  the Schr\"odinger equation
\begin{equation*}
 i \hbar \partial _t \Psi = - \hbar^2 /2    \Delta  \Psi + \Psi \, V
\end{equation*}
via 
\[  \Psi(t,x) = \sqrt {\mu(t,x)} e^{\frac i \hbar  \bar S(x,t) }
\]
where
\[ \bar S(x,t) = S(x,t) + \int_0^t L_F(S_\sigma,\mu_\sigma)d\sigma  \]
and $S(x,t) $ is the velocity potential of the flow $\mu$, i.e. satisfying $\int _M S d\mu =0$ and $\dot \mu_t = -\div (\nabla S_t\mu)$.
}
\end{corollary}

 \begin{remark}\label{bems} {\normalfont
 An equivalent version of theorem \ref{mainthm} puts  $\Psi = \sqrt
\mu(x,t) e^{ \frac i \hbar  S(x,t)}$ where $t\to (-\div (\nabla S_t
\mu_t), \mu_t)$ is a Lagrangian flow for $L_F$ and  $S$ is chosen to
satisfy for all  $t\geq 0$ \[\langle S_t,\mu_t\rangle - \langle
S_0,\mu_0\rangle = \int_0^t L_F(\dot \mu_s) ds. \]}
\end{remark}

\section{\normalsize \normalfont HAMILTONIAN STRUCTURE OF THE MADELUNG FLOW ON $T\P(M)$}
In this section we show that the Madelung flow \eqref{madelungflow}
has a Hamiltonian structure w.r.t. the canonical symplectic form
induced from the Wasserstein metric tensor on the tangent bundle
$T\P(M)$. To this aim we  use the representation
\[T\P(M) = \{ -\div(\nabla f\mu) \,|\, f\in C^\infty (M) , \mu \in \P(M)\}.\]


\begin{definition}[Standard Vector Fields on $T\P(M)$]
\label{maxvectorfieldef} Each pair $(\psi, \phi) \in C^\infty(M)
\times C^\infty(M)$ induces a vector field $V_{\phi,\psi}$ on
$T\P(M)$ via
\[ V_{\psi,\phi}(-\div(\nabla f \mu)) = \dot \gamma \]
where $t\to \gamma^{\psi,\phi}(t)=\gamma(t) \in T\P(M)$ is the curve
satisfying
\begin{gather*}
 \gamma(t) = -\div(\mu(t)  \nabla (f + t\phi)) \\
\mu_t = \exp(t\nabla \psi )_* \mu
\end{gather*}
\end{definition}

Recall that the standard symplectic form  on the tangent bundle of a
Riemannian manifold is given by $\omega = d\Theta$, where the
canonical 1-form $\Theta$ is defined as
\[\Theta(X) =  \langle \xi ,\pi_*(X)\rangle_{T_{\pi \xi}}, \quad  X \in T_\xi(TM),\]
and where $\pi$ denotes the projection map $\pi: TM \to M$. 

\begin{proposition} \label{formcomputation}  Let $\omega_\W \in \Lambda^2(T\P (M))$ be the standard symplectic form associated to the Wasserstein Riemannian structure on $\P(M)$, then
\begin{equation}
 \omega_\W (V_{\psi,\phi}, V_{\tilde \psi, \tilde \phi})
(-\div(\nabla f\mu )) = \langle \nabla \psi, \nabla \tilde \phi
\rangle_\mu -\langle \nabla \tilde \psi , \nabla \phi \rangle_\mu
\label{symplecticform} \end{equation}
\end{proposition}

\begin{proof} We use the formula
\begin{equation}
 \omega_\W (V_{\psi,\phi}, V_{\tilde \psi, \tilde \phi}) =
V_{\psi,\phi} \Theta (V_{\tilde \psi, \tilde \phi}) - V_{\tilde \psi
,\tilde \phi} \Theta (V_{ \psi, \phi}) - \Theta ([ V_ {\psi, \phi},
V_ {\tilde \psi,\tilde  \phi}]), \label{lietyp} \end{equation}
 where
$[ V_ {\psi, \phi}, V_ {\tilde \psi,\tilde  \phi}]$ denotes the
Lie-bracket of the vector fields $ V_ {\psi, \phi}$ and $ V_ {\tilde
\psi,\tilde \phi}$. From the definition of $\Theta$ we obtain \[
\Theta  ( V_ {\tilde \psi,\tilde  \phi}) (-\div(\nabla f \mu )) =
\langle \nabla f, \nabla \tilde \psi \rangle_\mu. \]
 Hence
\begin{equation}
\begin{split}
V_ { \psi,  \phi}(\Theta  ( V_ {\tilde \psi,\tilde  \phi}) )  & = \frac d{dt}_{|t=0} \Theta  ( V_ {\tilde \psi,\tilde  \phi}) ( \gamma^{\psi,\phi}(t)) \\
& = \frac d{dt}_{|t=0}  \langle \nabla (f+t \phi), \nabla \tilde \psi \rangle_{\mu(t)}  \\
 & =  \langle \nabla \phi, \nabla \tilde \psi \rangle_\mu - \int _M \nabla f \cdot  \nabla \tilde \psi  (-\div \nabla \psi \mu ) dx \\
&   =  \langle \nabla \phi, \nabla \tilde \psi \rangle_\mu + \int _M
\nabla (\nabla f \cdot  \nabla \tilde \psi )   \nabla \psi d\mu
\end{split} \label{firstpart}
\end{equation}
Next, since $\Theta$ measures tangential variations only one gets
that \begin{equation}
 \Theta( [ V_ {\psi, \phi}, V_ {\tilde
\psi,\tilde \phi}])(-\div(\nabla f \mu)) = \langle \nabla f ,
[\nabla \psi, \nabla \tilde \psi ]\rangle_\mu. \label{secpart}
\end{equation} Finally, it is easy to check that
\[\int _M \nabla (\nabla f \cdot  \nabla \tilde \psi )   \nabla \psi d\mu -\int _M \nabla (\nabla f \cdot  \nabla \psi )   \nabla \tilde \psi d\mu  - \langle \nabla f , [\nabla \psi, \nabla \tilde \psi ]\rangle_\mu  =0,\]
which together with \eqref{lietyp}, \eqref{firstpart} and
\eqref{secpart} establishes the claim. \bbox
\end{proof}

\begin{remark} {\normalfont Proposition \ref{formcomputation}   shows that
$\omega_\W$  is the lift of the standard symplectic form on $TM$ to
$T\P(M)$. This corresponds to the result   in \cite[section
6]{MR2358290}, which however is less explicit  than formula
\eqref{symplecticform}. }
\end{remark}

Using the
 the Riemannian inner product in each fiber of $T\P(M)$ the Hamiltonian associated with $L_F$ is
\begin{equation}
 H_F: T\P(M) \to \R; \quad H_F(-\div (\nabla f \mu)) = \frac 1 2 \int_M |\nabla f|^2 d\mu + F(\mu)
 \label{hamiltonian}
\end{equation}

\begin{proposition}{\it  Let $X_F$ denote the Hamiltonian vector field $X_F$  induced on $T\P(M)$
from $H_F$ and $\omega_\W$, then 
\[ X_F(-\div (\nabla f \mu)) = V_{f, - (\frac 1 2 |\nabla f|^2 + V +
\frac {h^2}{8 } (|\nabla \ln \mu |^2 -2 \frac {\Delta \mu}{\mu}))} (
-\div (\nabla f  \mu ))\] }
\end{proposition}

\begin{proof}
Fix  $\psi, \phi \in C^\infty(M)$ and let  $V_{\psi, \phi}(.)$
denote the corresponding standard vector field. Let $t \to \gamma
(t) = -\div((\nabla f +t \phi) \mu_t )$, where $\mu_t =\exp(t \nabla
\psi )_*\mu$,  denote the corresponding curve on $T\P(M)$, then
\begin{align*} V_{\psi,\phi}(H_F)(-\div(\nabla f \mu)) & = \partial
_{t | t=0} H_F(\gamma(t))\\
& =  \partial _{t | t=0}  \bigl( \frac 1 2 \int_M |\nabla
(f+t\phi)|^2 d\mu_t + \langle V,\mu_t\rangle + \frac {h^2} 8 I
(\mu_t)\bigr)\\ & = I + II + III,
\end{align*}
where
\begin{align*}
I &= \int_M \nabla f \nabla \phi d\mu + \frac 1 2 \int_M |\nabla f
|^2 (-\div(\nabla \psi \mu )) \\
& = \langle \nabla f ,\nabla \phi \rangle_\mu + \langle \nabla \psi
, \nabla (\frac 1 2 |\nabla f|^2) \rangle \end{align*}
\begin{align*}
II &= \int_M V (-\div(\nabla \psi \mu ))  = \langle \nabla V ,
\nabla \psi \rangle_\mu \end{align*}
and
\begin{align*}
III &= \frac {\hbar ^2 }{8} \int_M 2 \nabla \ln \mu_t \nabla (\frac
{-\div(\nabla \psi \mu )}{\mu}) d\mu + \frac {\hbar ^2 }{8} \int_M
|\nabla \ln\mu |^2 (-\div(\nabla \psi \mu)) \\
& = \frac {\hbar ^2 } 8 \bigl( \langle \nabla \psi, \nabla ( - \frac
{2 \Delta \mu }{\mu})\rangle_\mu + \langle  \nabla \psi , \nabla
|\nabla \ln \mu|^2 \rangle_\mu \bigr)
\end{align*}

Hence, collecting terms
\[ V_{\psi,\phi}(H_F)(-\div(\nabla f \mu)) = \langle \nabla f,
\nabla \phi \rangle_\mu - \langle \nabla(-(\frac 1 2 |\nabla f|^2 +
V + \frac {\hbar ^2 } 8 (|\nabla \ln \mu |^2 -2  \frac {\Delta \mu
}{\mu}))), \nabla \psi\rangle_\mu.\]

From this and formula \eqref{symplecticform} the claim follows.
\bbox
\end{proof}

\begin{corollary} The pair $t \to (S_t,\mu_t) \in C^\infty(M)\times
\P(M)$ solves the Madelung flow equation \eqref{madelungflow} if and
only if $t \to -\div(\nabla S_t \mu_t) \in T\P(M)$ is an integral
curve for $X_F$.
\end{corollary}

\section{\normalsize \normalfont THE MADELUNG TRANSFORM AS A SYMPLECTIC SUBMERSION}

In this section we prove that the two equations
\eqref{eulerlagrange} and \eqref{seq} are related via a symplectic
submersion.
\begin{definition} \label{submersiondef} A   smooth map  $s: (M,\omega) \to (N,\eta)$ between
two symplectic manifolds is called a symplectic submersion if its
differential $s_*: TM \to TN$ is surjective and satisfies
$\eta(s_*X, s_*Y) = \omega(X,Y)$ for all $X,Y \in
TM$.\end{definition}

Note that this definition implies in particular that the map $s$
itself is surjective. The following proposition is easily verified.
Its meaning is that in order to solve a Hamiltonian system on $N$ we
may look for solutions for the lifted Hamiltonian $g \circ s$ on the
larger state space $M$ and project them via $s$ back again to $N$.

\begin{proposition} \label{liftprop} Let $s: (M,\omega) \to (N,\eta)$ be a symplectic
submersion and let $f \in C^\infty(M)$ and $g\in C^\infty(N)$ with
$g\circ s = f$, then $s$ maps Hamiltonian flows associated to $f$ on
$(M, \omega)$ to Hamiltonian flows associated to $g$ on $(N, \eta)$.
\end{proposition}

Let now $\C(M)= C^\infty(M;\mathbb C)$ denote the linear space of
smooth complex valued functions on $M$. Identifying as usual the
tangent space over an element $\Psi \in \C$ with  $\C$, $T\C$ is
naturally equipped with the symplectic form
\[ \omega_\C(F, G) = - 2\int _M \Im(F\cdot  \ol G) (x)
dx.\] It is a well-known  fact that the Schr\"odinger equation
$\eqref{seq}$ is the Hamiltonian flow induced from the symplectic
form $\hbar \cdot \omega_C$ and the Hamiltonian function on $\C$
\[H_S(\Psi)
= \frac {\hbar^2 } 2 \int_{M} |\nabla \Psi |^2 dx + \int_M
|\Psi(x)|^2 V(x) dx.\]

Let $\C_*(M)$ denote the subset of nowhere vanishing functions from
$\C$ such that $\int_ M |\Psi(x)|^2 dx =1$ and note that $\C_*(M)$
 is invariant under the Schr\"odinger flow.

\smallskip Assuming simple connectedness of $M$ implies (via a
standard lifting theorem of algebraic topology)  that each function
$\Psi \in \C_*$ admits a decomposition $\Psi=  |\Psi| e^{\frac i
\hbar S}$, where the smooth field $S: M \to \R$ is uniquely defined
up to an additive constant $\hbar 2\pi k$, $k\in \N$. Hence we may
define a the \textit{Madelung transform}\begin{equation}  \sigma:
\C_*(M) \to TP(M), \qquad
 \sigma(\Psi) = - \div(
|\Psi|^2 \nabla S  ).\label{madtrans}
\end{equation}

For the next theorem recall that in our definition of $T\P(M)$ we
assume that the supporting measures are smooth and strictly positive
on $M$.

\begin{theorem} \label{projectthm} Let $M$ be simply connected. Then the  Madelung transform
\[\sigma: \C_*(M) \to T\P(M), \qquad \sigma( |\Psi| e^{\frac {i}{\hbar }
S}) = -\div(|\Psi|^2 \nabla S) \]
  defines  symplectic submersion from  $(\C_*(M),\hbar \cdot \omega_\C)$ to
  $(T\P(M),\omega_\W)$ which preserves the Hamiltonian, i.e.\ \[H_S =
  H_F\circ \sigma.\]\end{theorem}

\begin{remark}{\normalfont Together with proposition \ref{liftprop}
this result presents the Schr\"odinger equation \eqref{seq} as a
symplectic lifting of Newton's law on Wasserstein space
\eqref{eulerlagrange} to the larger space $\C_*(M)$, and which can
be solved much easier because it is  linear. Projecting the solution
down to $T\P(M)$ via $\sigma$ yields the desired solution to
\eqref{eulerlagrange}. Going in inverse direction from
\eqref{eulerlagrange} to \eqref{seq} requires a scalar correction
term in the phase field, c.f. remark \ref{factorize}.}
\end{remark}

\textit{Proof of theorem \ref{projectthm}.} Clearly,
$\sigma(C_*(M))= T\P(M)$. To see that  $\sigma: C_*(M) \to T\P(M)$
is a submersion fix a reference point $0\in M$, then for each $r\in
[0, 2\pi\hbar [$ the map $\tau=\tau^{(r)}$
\[\tau  : T\P(M)  \to
C_*(M), \qquad -\div(\nabla  S \mu) \to \sqrt {\mu} e^{\frac i \hbar
(S-(S(0)-r)) },\] is a bijection from $T\P(M)$ to the subset $\{
\Psi \in C_*, \frac{\Psi}{|\Psi|}(0) = e^{\frac i {\hbar} r}\}$
which satisfies $\sigma \circ \tau = \mathop{\normalfont
Id}_{T\P(M)}$. This proves that the differential $s_*$ of $s$ is
surjective.

\smallskip To prove that $\sigma$ is symplectic  let  $\Psi = \sqrt
\mu e^{\frac i \hbar f}\in \C_*$ with $f(0)=r \in [0,2\pi \hbar[$
and let $\eta = - \div( \mu \nabla f) = \sigma(\Psi)\in T\P(M)$.
Again due to  the  identity $\sigma \circ \tau = \mathop{\normalfont
Id}_{T\P(M)}$  it suffices to prove that $\tau^* \omega _\C =
1/\hbar \cdot \omega_\W$ on $T_\eta (T \P(M))$. Since the set $
\{V_{\psi,\phi}(-\div(\mu \nabla f))\,| \psi,\phi \in C^\infty
(M)\}$ spans the full tangent space $T_\eta (T \P(M))$, it remains
to  verify
\[\omega_\C ( \tau_* V_{\psi,\phi}, \tau_* V_{\tilde \psi,\tilde
\phi}) = \frac 1 \hbar \omega_\W(  V_{\psi,\phi},   V_{\psi,\phi})\]
for all $\psi, \phi, \tilde \psi, \tilde \phi  \in C^\infty(M)$. By
definition of $V_{\psi,\phi}$ and $\tau=\tau^{(r)}$ for {  $\mu_t :=
\exp (t \nabla \psi)_*(\mu)$ and  $c(t) :=  f(0)+ t\phi(0)-r$}
\begin{align*} \tau_*
V_{\psi,\phi}  & =
\partial _{t | t=0} \sqrt{\mu_t } e^{\frac i \hbar  (f +t\phi- c(t))}
  = e^{\frac i \hbar f} \left (
\frac  1 {2\sqrt \mu } ( -\div (\nabla \psi \mu ) ) + \sqrt
{\mu}\frac i \hbar  (\phi - \dot c ) \right)
\end{align*}
Hence
\begin{align*}
\omega_\C(\tau_* V_{\psi,\phi}, \tau_* V_{\tilde \psi,\tilde \phi})
& = - 2 \int_M \bigl ( \frac 1 {2\sqrt \mu}  (-\div (\nabla \psi
\mu)) \cdot ( - \sqrt \mu \frac 1 \hbar  (\tilde \phi
  + \dot{\tilde c})) \\  & \phantom{= - 2 \int_M \bigl (} + \sqrt \mu \frac 1 \hbar  (\phi + \dot c ) \cdot
\frac  1 {2  \sqrt \mu }
(-\div ( \nabla \tilde \psi \mu )) \bigr) dx \\
& = \frac 1 \hbar  \bigl( \langle  \nabla \psi  ,  \nabla \tilde
\phi \rangle_\mu   - \langle \nabla \phi ,  \nabla \tilde \psi_\mu
\rangle \bigr) = \frac 1 \hbar  \omega_\W (V_{\psi,\phi},V_{\tilde
\psi,\tilde\phi})
\end{align*}

Finally, for $\Psi = \tau (-(\div \nabla f \mu))$,   $\nabla \Psi  =
\sqrt {\mu } e^{\frac i \hbar f} ( \frac 1 2 \nabla \ln \mu   +
\frac i{\hbar} \nabla f)$ such that
\[ \frac {\hbar^2  }{2}\int_M |\nabla \Psi|^2 = \frac 1 2  \int_M |\nabla f |^2 d\mu +
\frac{\hbar^2 }{8}I(\mu)\] and $\int |\Psi (x)|^2 V (x) dx = \langle
V, \mu\rangle$ which establishes the third claim $H_S = H_F\circ
\sigma$ of the theorem. \bbox


\section{\normalsize \normalfont APPENDIX - BASIC FORMAL RIEMANNIAN
CALCULUS ON $\P(M)$}\label{append} Let $\P_2(M)$ denote the set of
Borel probability measures $\mu$ on a smooth closed finite
dimensional Riemannian manifold  $(M,g)$  having finite second
moment $\int_M d^2(o,x) \mu(dx) < \infty$. As argued in
\cite{MR2358290} the subsequent calculations  make strict
mathematical sense on the $d_\W$-dense subset of smooth fully
supported probabilities $\P^\infty(M)\subset \P_2(M)$ which shall
often be identified with their corresponding density $\mu
\stackrel {\wedge}{=} d\mu/dx$. \smallskip \smallskip

 \textit{Vector Fields on $\P(M)$ and Velocity Potentials.}
 \smallskip

 A function $\phi \in \C_c^\infty(M)$ induces a flow on $\P(M)$ via push forward
\[   t \to \mu_t = (\Phi^{\nabla \phi }_t) _* \mu_0,
\]
where $t \to \Phi_t $ is the local flow of difformorphisms on $M$
induced from the vector field $\nabla \phi \in \Gamma(M)$ starting
from $\Phi_0 = \rm{Id}_M$. The continuity equation yields the
infinitesimal variation of $\mu \in \P(M)$  as
\[  \dot \mu = \partial_{t|t=0}\mu_t = - \div (\nabla \phi \mu) \in T_\mu(\P).
\]
Hence the function $\phi$ induces a vector field $V_\phi \in \Gamma(\P(M))$ by
\[  V_\phi (\mu) = -\div(\nabla \phi \mu),
\]
acting on smooth functionals $F: \P(M) \to \R$ via
\[
 V_\phi(F)(\mu) = \partial _{\epsilon | \epsilon =0 } F(\mu- \epsilon \div(\nabla \phi \mu)) = \partial _{t | t =0} F((\Phi^{\nabla \phi}_t)_*\mu)
\]
with Riemannian norm
\[\norm{ V_\phi(\mu)}_{T_\mu\P} ^2 = \int_M |\nabla \phi|^2(x) \mu(dx).\]
Conversely, each smooth variation $\psi \in T_\mu(\P)$ can be identified with
\[  \psi = - V_\phi(\mu) \quad \mbox{ with } \phi = G_\mu \psi,\]
 where $G_\mu$ is the Green operator for $\Delta^\mu : \phi \to -\div(\mu \nabla \phi )$ on $L^2_0(M, dx)=L^2_0(M, dx)\cap\{ \langle f,dx\rangle =0\}$. Hence, for each $\psi \in T_\mu \P$ there exists a unique $\phi \in \C^\infty \cap L^2(M, dx)$ such that
\[ \psi = - \div (\mu \nabla \phi ) \mbox { and } \langle \phi, \mu \rangle =0,
\]
which we call  velocity potential  for $\psi \in T_\mu\P(M)$.

\smallskip

\textit{Riemannian Gradient on $\P(M)$.} \smallskip

 The Riemannian gradient of
a smooth functional $F : {\rm Dom}(F) \subset \P(M) \to \R$ is
computed to be
\[  \nabla ^\W F_{|\mu} = - \Delta ^\mu (DF_{|\mu}),
\]
where $x\to DF_{|\mu}(x)$ is the $L^2(M, dx)$-Frechet-derivative of $F$ in $\mu$, which is defined through the relation
\[
 \partial _{\epsilon|\epsilon =0} F(\mu+ \epsilon \xi) =\int_{M}  DF_{\mu}(x)\xi(x) dx,
\]
for all $\xi$ chosen from a suitable dense set of test functions in $L^2(M,dx)$.
The following  examples are easily obtained.
 \[
\begin{array}{ll}
   F(\mu) = \int_M \phi(x) \mu(dx),  & \nabla^\W F_{|\mu} = V_\phi(\mu)= -\div(\nabla \phi \mu)\\
F(\mu) = \int_M \mu \log \mu  dx,  & \nabla^\W F_{|\mu} = -\div(\mu \nabla \log \mu)= - \Delta \mu \\
  F(\mu) = \int_M  |\nabla \ln\mu|^2d\mu, & \nabla^\W F_{|\mu}= - \div(\mu \nabla (|\nabla \ln \mu|^2 - \frac 2 \mu \Delta \mu)).
\end{array}
\]
Here $\Delta$ denotes the Laplace-Beltrami operator on $(M,g)$. As a
consequence, the Boltzmann entropy induces the heat equation as
gradient flow on $\P(M)$, and the information functional is the
norm-square of its gradient, i.e.  \[ \norm {\nabla ^\W
\Ent_{|\mu}}^2_{T_\mu\P} =  \norm{-\div(\mu \nabla \log
\mu)}^2_{T_\mu\P} = \int_M |\nabla \log \mu |^2 d\mu = I(\mu).\]

\textit{Covariant Derivative.}  \smallskip

 The Koszul identity for
the Levi-Civita connection and a straightforward computation of
commutators  show \cite{MR2358290} for the covariant derivative
$\nabla^\W$ associtated to $d_\W$ that
\[
 \langle \nabla ^\W _{V_{\phi_1}} V_{\phi_2}, V_{\phi_3}\rangle_{T_\mu} = \int_M \Hess \phi_2 (\nabla \phi_1, \nabla\phi_2) d\mu.
\]
For a smooth curve $t\to \mu(t)$ with $\dot \mu_t = V_{\phi_t}$ this yields
\[
\nabla ^\W _{\dot \mu } \dot \mu =  V_{\partial _t \phi + \frac 1 2 |\nabla \phi|^2}.
\]

\end{document}